\documentclass[letterpaper, 10 pt, conference]{ieeeconf}
\IEEEoverridecommandlockouts 
\usepackage[utf8]{inputenc} 
\usepackage[T1]{fontenc}
\usepackage{url}
\usepackage{ifthen}
\usepackage{cite,color}
\usepackage[cmex10]{amsmath} 
                             
\newcommand{\stirlingii}{\genfrac{\{}{\}}{0pt}{}}
\usepackage[final]{graphicx}

\usepackage{amsthm}
\usepackage{amssymb}
\usepackage{euscript}
\usepackage{subfigure}
\usepackage[colorinlistoftodos]{todonotes}

\title{\LARGE \bf On the Effect of Task-to-Worker Assignment in Distributed Computing Systems with Stragglers} 

\author{Amir Behrouzi-Far$^{1}$ and Emina Soljanin$^{2}$
\thanks{$^{1}$Amir Behrouzi-Far is with the Department of Electrical and Computer Engineering,
        Rutgers University, Piscataway, NJ 08854, USA
        {\tt\small amir.behrouzifar@rutgers.edu}}%
\thanks{$^{2}$Emina Soljanin is with the Department of Electrical and Computer Engineering, Rutgers University,
        Piscataway, NJ 08854, USA
        {\tt\small emina.soljanin@rutgers.edu}}%
}

\begin{document}

\maketitle
\thispagestyle{empty}
\pagestyle{empty}

\newtheorem{theorem}{Theorem}
\newtheorem{definition}{Definition}
\newtheorem{lemma}{Lemma}
\newtheorem{proposition}{Proposition}

\begin{abstract}
We study the expected completion time of some recently proposed algorithms for distributed computing which redundantly assign computing tasks to multiple machines in order to tolerate a certain number of machine failures. 
We analytically show that not only the amount of redundancy  but also the task-to-machine assignments affect the latency in a distributed system. 
We study systems with a fixed number of computing tasks that are split in possibly overlapping batches, and independent exponentially distributed machine service times. 
We show that, for such systems, the uniform replication of non-overlapping (disjoint) batches of computing tasks achieves the minimum expected computing time.
   
\end{abstract}

\section{Introduction}

Distributed computing has gained great attention in the time of big data \cite{kshemkalyani2011distributed}. By enabling parallel task execution, distributed computing systems can bring considerable speed ups to e.g. matrix multiplication \cite{choi1998new}, model training in machine learning \cite{dean2012large} and convex optimization \cite{boyd2011distributed}.
However, implementing computing algorithms in a distributed system introduces new challenges that have to be addressed in order to benefit from paralleziation. In particular, since the failure rate and/or slow down of the system increase with the number of computing nodes, robustness is an essential part of any reliable distributed computing/storage algorithm \cite{elser2005reliable}. For achieving robustness, redundant computing/storage is introduced in the literature \cite{elser2005reliable,brun2011smart}. Redundancy in a distributed computing system enables the system to generate the overall result form the computations of a subset of all computing nodes, which provides robustness to failed and/or slow nodes, known as stragglers.

A number of algorithms for distributed computing which tolerate the failure of some computing nodes have recently appeared in the literature. For example, in \cite{tandon2016gradient} used error correcting codes and, for distributed gradient descent, authors established lower bound on the degree of redundancy for tolerating a specific number of node failures, and introduced some data-to-node assignment methods that achieve the lower bound. In \cite{halbawi2017improving}, Reed-Solomon codes are studied for distributed gradient descent, and the code construction, decoding algorithm and an asymptotic bound on the computation time are provided. In \cite{li2017near}, a random data assignment to computing nodes is proposed, and claimed to outperform the deterministic method in \cite{tandon2016gradient} in terms of average per-iteration computing time. Fault-tolerant distributed matrix multiplication is addressed in e.g. \cite{yu2017polynomial}. In this line of work, authors focus on failure tolerance capability of their proposed algorithms without considering the resulting computing time. 

The effect of the replicated redundancy on the distributed computing latency was studied in \cite{wang2015using}, where it is shown that delayed replication of the straggling tasks could reduce both cost and latency in the system. Moreover, authors in \cite{aktas2017effective} analyzed the latency of the system with coded redundancy, where instead of simple replication of the straggling tasks a coded combination of them would be introduced. The analysis of the same work shows that, coded redundancy could achieve the same latency with less cost of the computing. In \cite{aktas2018straggler} authors show that relaunching straggling tasks could significantly reduce the cost and latency in a distributed computing system. In these works, the delay of computing is considered without looking into the effect of task assignment in the timing performance.

In this paper, we study the computing time of the failure-tolerant distributed computing algorithms. We specifically focus on the algorithms proposed by \cite{tandon2016gradient} and \cite{li2017near}, which used error correcting codes and random task assignment respectively, and show that how the task assignment could improve the expected computing time of the same methods. We consider a distributed computing system consisting of multiple worker nodes, which perform a specific computation over a sebset of a possibly huge data set, and a master node, which collects the local computations from worker nodes and generates the overall result. The original data set is (redundantly) distributed among worker nodes and the subset of data at each worker is called a \textit{data batch}. Through analysis, we show that, with a fixed number of computing tasks and exponentially distributed service time of worker nodes, the lower bound of the expected computing time among different data distribution policies is achieved by the balanced assignment of non-overlapping batches. We also derived this lower bound and verified it by simulations.

The organization of this paper is as follows. In section \ref{sectionII}, we introduce the model for system architecture, computing task, data distribution and service time of the worker nodes. The data distribution policies, non-overlapping and overlapping batches, are described in section \ref{sectionIII}. In section \ref{sectionIV} we provide analysis of the computing time for each data distribution policies. The numerical results and the concluding remarks are provided in section \ref{sectionV} and section \ref{sectionVI}, respectively.

\section{System Model}\label{sectionII}

\subsection{Distributed Computing Architecture}
We study the system given in Fig. \ref{fig:arc}, which will be referred to as System~\ref{fig:arc} in the rest of the paper. The system consists of a data set with $S$ data blocks, a batching unit which partitions the data blocks into $B$ batches, a batch assignment unit which allocates data batches among $N$ worker nodes, each performing a specific computing task on its data batch, and a master node which collects local (coded) computations from worker nodes and computes the overall result accordingly. For a given data set, a fixed number of identical worker nodes and a master node, the average computing time of the system depends on how the data is distributed among workers, what the service time model of the worker nodes is, and how each worker codes its computations before sending it to the master node.
\begin{figure}[htbp]
   \centering
   \includegraphics[width=\columnwidth, keepaspectratio]{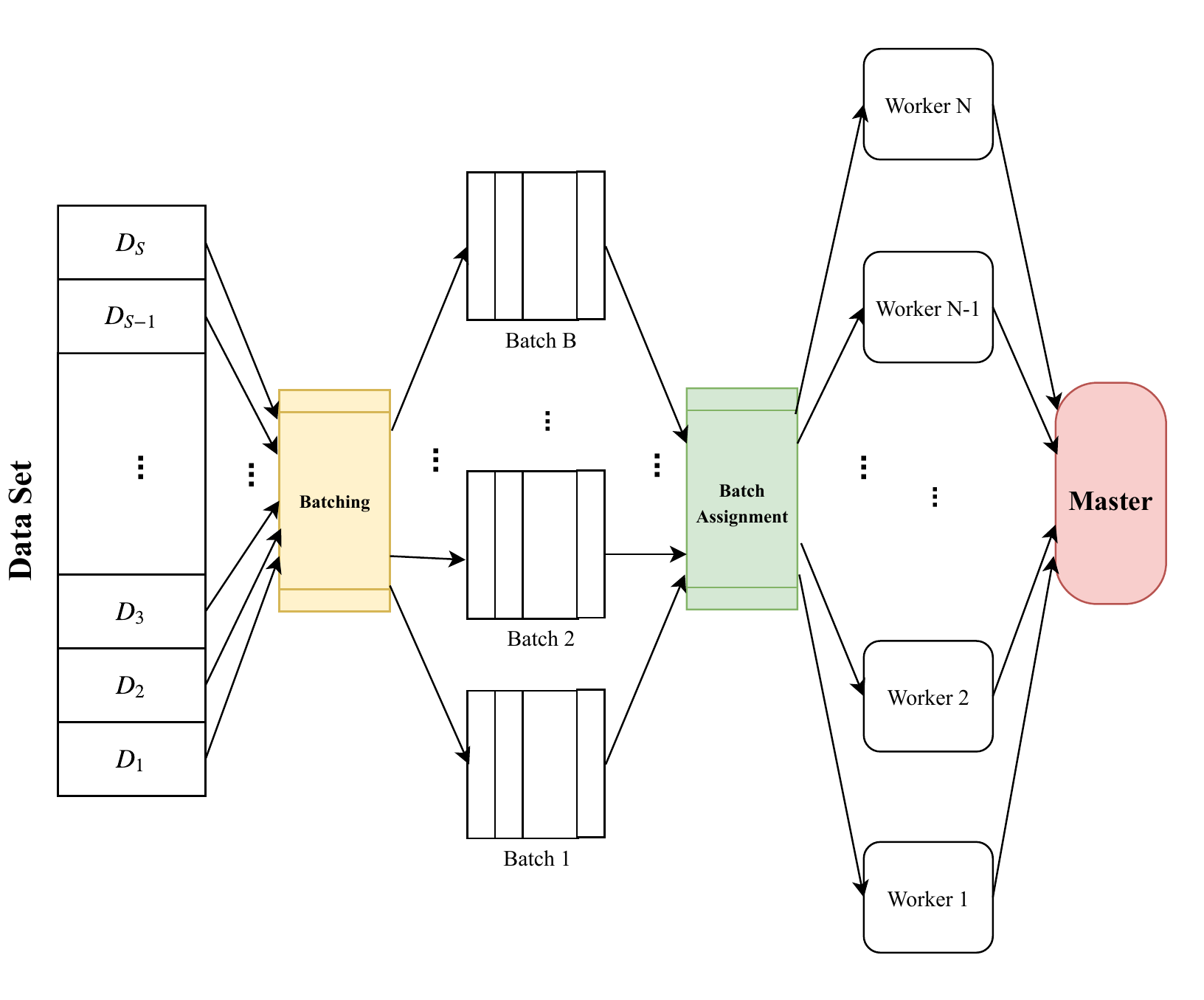}
   \caption{System model.}
   \label{fig:arc}
\end{figure}
\subsection{Computing Task Model}
In this work, we consider data-intensive jobs that could be divided into smaller tasks to be executed in parallel. Each task is assigned to a worker node which performs computations over a subset of the entire data and sends its local result to the master node. After receiving the local results from a large enough number of workers, the master node generates the overall result. This is a well applicable and widely used model for different algorithms, e.g. matrix multiplication \cite{choi1998new} and gradient based optimizers \cite{tandon2016gradient,boyd2011distributed}. As a simple example of this computing model, consider a set of natural numbers $D=\{1,2,3,\dots,z\}$. Suppose we are interested in the sum of $\ell$-th power of the elements of this set, i.e.,
\begin{equation*}
S_z^{(\ell)}=\sum_{k=1}^{z}k^\ell.  
\end{equation*}
This summation could be decomposed into $N$ sums (one for each worker):
\begin{equation}
S_z^{(\ell)}=\sum_{k=1}^{z_1}k^\ell+\sum_{k=z_1+1}^{z_2}k^\ell+
\sum_{k=z_2+1}^{z_3}k^\ell
+\dots+\sum_{k=z_{N-1}+1}^{z_N}k^\ell,
\label{computing model}
\end{equation}
where $z_1<z_2<z_3<\dots<z_{N-1}<z_N=z$. Now, if we split the set $D$ into subsets of size $z_1,z_2-z_1,z_3-z_2,\dots,z_{N}-z_{N-1}$ and assign each subset to one worker node, then a worker would have to compute the $\ell$-th power for each number in the data subset (batch) assigned to it and add them up. Then the master node would be able to generate $S_z^{(\ell)}$ by summing the local computation results.

\subsection{Data Distribution}
To ensure reliable computations, each block of the data set is stored redundantly among workers. The redundant distribution of data among the worker nodes is abstracted into a two-stage process. In the first stage, the data set is stored in equal-sized batches and in the second stage batches get assigned to the workers. Note that, the batches could be non-overlapping, when the data set is simply chopped into smaller parts, or they could overlap, if each data block is available in more than one batch. Nevertheless, the batch sizes will be the same in both cases. Therefore, there would be more batches if they overlap. Assuming the batch size is $S/B$, where $B$ divides $S$, the number of batches is an integer in the range of $[B,N]$.
\subsection{Service Time Model}
We define $T_{i,j}$ as the time that worker $j$ takes to perform computing over data batch $i$ and communicate the result to the master node. We assume that $T_{i,j}$ are all independent and exponentially, identically distributed:
\begin{equation}
    T_{i,j} \sim \exp(\lambda),
    \label{exp}
\end{equation}
where $\lambda$ is the service rate of the worker nodes and is identical for all of the workers.

\section{Data Distribution}\label{sectionIII}
The data distribution policies could be categorized as: $1)$ non-overlapping batches, or $2)$ overlapping batches. We will discuss each policy in detail in the following.
\subsection{Non-overlapping Batches}
Under this batching policy, the data set of $S$ blocks is split into $B$ equal-size batches, giving batch size $S/B$. When $B<N$, batches can be assigned redundantly to the $N$ worker nodes, in order to enable failure and/or straggler tolerance in the system. Since the intersection of the batches are empty, the data at each worker either completely overlap or do not overlap at all with the data at any other worker. Therefore, for acquiring the computations over a specific batch, the master node should receive the local results from at least one of the workers hosting that batch. 

The random batch assignment in which each worker draws a batch, uniformly at random, with replacement from the pool of batches was studied in \cite{li2017near}, and can be naturally modeled as the
 Coupon Collection (CC) problem. It was shown in \cite{li2017near} that the random assignment reduces the computation time per iteration compared to deterministic assignment in \cite{tandon2016gradient}. We will show in the following section that the imbalance data distribution among workers, which results from the random assignment, adversely affects the computation time.  On the  other hand, since some of the batches will be drawn only once, there will be no failure tolerance for a worker if its data is not replicated at any other worker. Furthermore, by random assignment, there is always a non-zero probability that some batches not get selected, leading to an inaccurate computation result. 
 
Let $n$ be the number of workers for covering all the batches in random assignment policy. The following proposition provides the data coverage probability with the random batch-to-worker assignment.

\begin{proposition}
The probability of covering $B$ batches with $N$ workers with random batch-to-worker assignment is given by,
\begin{equation}
P(n\leq N) = \frac{B!}{B^N}\stirlingii{N}{B},
\label{covering}
\end{equation}
where $\stirlingii{n}{B}$ is the Stirling number of second kind \cite{weisstein2002stirling}, given by,
\begin{equation*}
\stirlingii{n}{B}=\frac{1}{B!}\sum_{i=0}^{B}(-1)^{B-i}\binom{B}{i}i^n.
\end{equation*}
\end{proposition}
\begin{proof}
The probability of covering $B$ batches with exactly $N$ workers is given in \cite{myers2006some}, as
\begin{equation}
P(n=N)=\frac{B!}{B^N}\stirlingii{N-1}{B-1}.
\end{equation}
Therefore,
\begin{equation*}
\begin{split}
P(n\leq N)&=\sum_{n=B}^{N}\frac{B!}{B^n}\stirlingii{n-1}{B-1}\\
          &=B!\sum_{n=B-1}^{N-1}\frac{1}{B^{n+1}}\stirlingii{n}{B-1}\\
          &=\frac{B!}{B^N}\sum_{n=B-1}^{N-1}B^{N-n-1}\stirlingii{n}{B-1}\\
          &=\frac{B!}{B^N}\stirlingii{N}{B},
\end{split}
\end{equation*}
where the last equation holds according to \cite{comtet2012advanced}.
\end{proof}

\begin{figure}[htbp]
   \centering
   \includegraphics[width=\columnwidth, keepaspectratio]{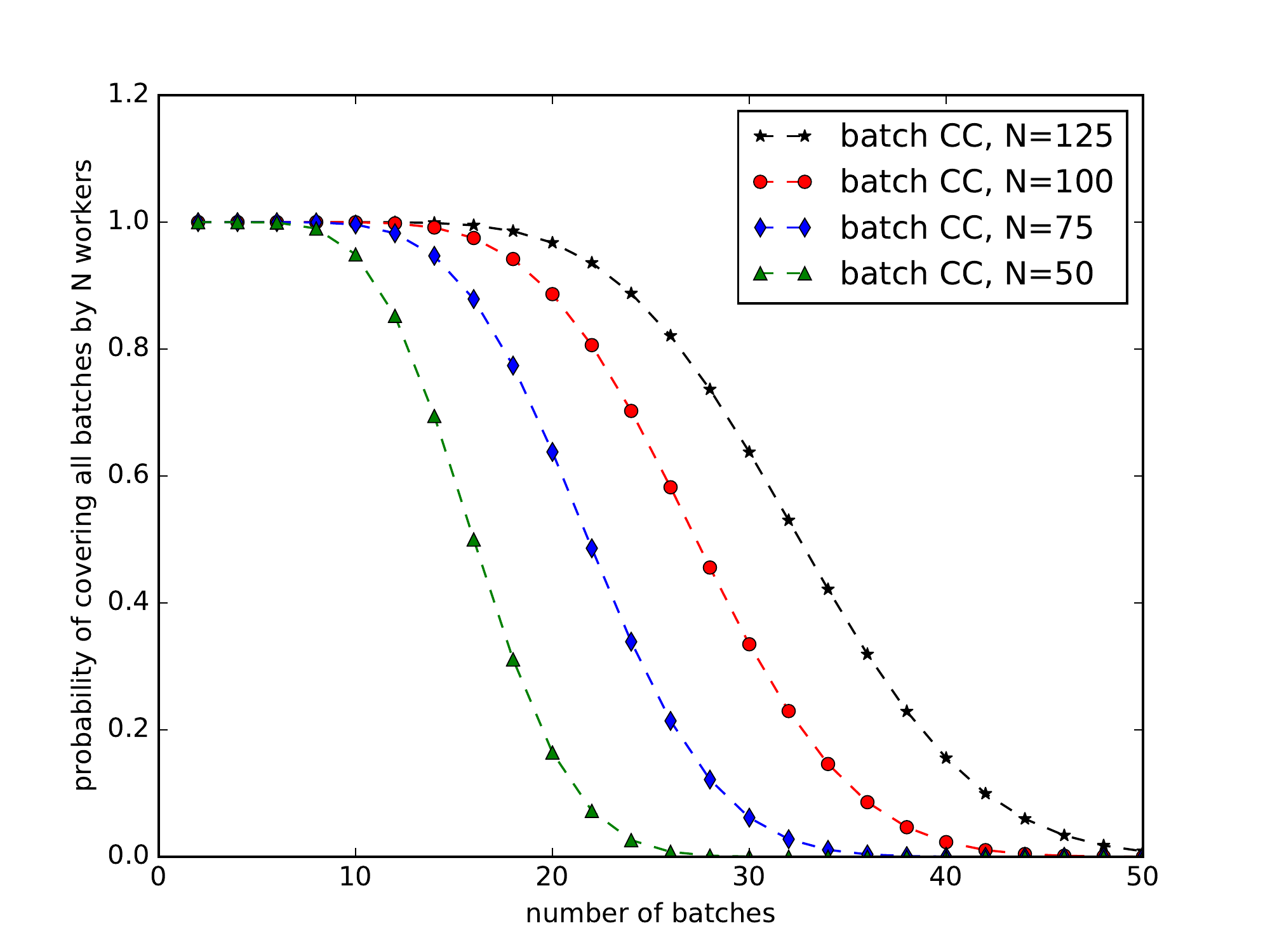}
   \caption{Coverage probability of batch CC with $N$ workers.}
   \label{fig:stirling}
\end{figure}

The probability of covering all batches (\ref{covering}) vs.\ the number of batches is plotted for four different values of $N$ in Fig.~\ref{fig:stirling}. We see that in order to cover all the $B$ batches with $N$ workers, the number of workers should be much larger than the number of batches. 
On the other hand, as the number of batches increases, the probability of covering them with a given number of workers decreases. Note that, the number of batches is of great interest to us, since file systems, \cite{ghemawat2003google} and \cite{shvachko2010hadoop}, usually store data in batches of fixed size. Therefore, the larger the data the larger the number of data batches. For example, the batch size in Hadoop distributed file system (HDFS) is 64MB, \cite{shvachko2010hadoop}.

\subsection{Overlapping Batches}
Under this batching policy, the data set is stored in $N$ overlapping batches, each assigned to a worker node. To be able to compare different data distributions' computing time, we keep the same batch sizes in both overlapping and non-overlapping policies. Thus, the batch size with overlapping batches is $S/B$, where $B$ is the number of batches in non-overlapping policy. Note that with overlapping batches, the number of batches is equal to the number of workers and, therefore, we need not decide about the number of workers assigned to each batch, which was the problem with non-overlapping batches. Besides, the challenge here is that which batches and how much should they overlap for faster computing. This question is in general very hard to answer (due to the huge size of the problem) and we will study it under the following assumption. We assume that the entire set of batches could be divided into non-overlapping \textit{group} of batches, such that each block of data set appears one and only one time in each group. In other words, each group hosts the entire data set, divided into batches of size $S/B$. Now the question is how we should arrange the elements of the data set in each group to have faster computations, which will be answered in section \ref{batching policy}.



\section{Computing Time Analysis}\label{sectionIV}
In this section we analyse the computing time of System~\ref{fig:arc}, under the two data distribution policies described in section \ref{sectionIII}.
\subsection{Computing Time with Non-overlapping Batches}
Let $N_i$ be the number of workers that are assigned with the data batch $i$, and define
$\EuScript{\Bar{N}}=(N_1,N_2,\dots,N_B)$ as the assignment vector. Recall that the service times of the worker nodes (the computation time and the time it takes to communicate the results) are assumed to be exponentially distributed with rate $\lambda$. 

\begin{figure}[htbp]
   \centering
   \includegraphics[width=\columnwidth, keepaspectratio]{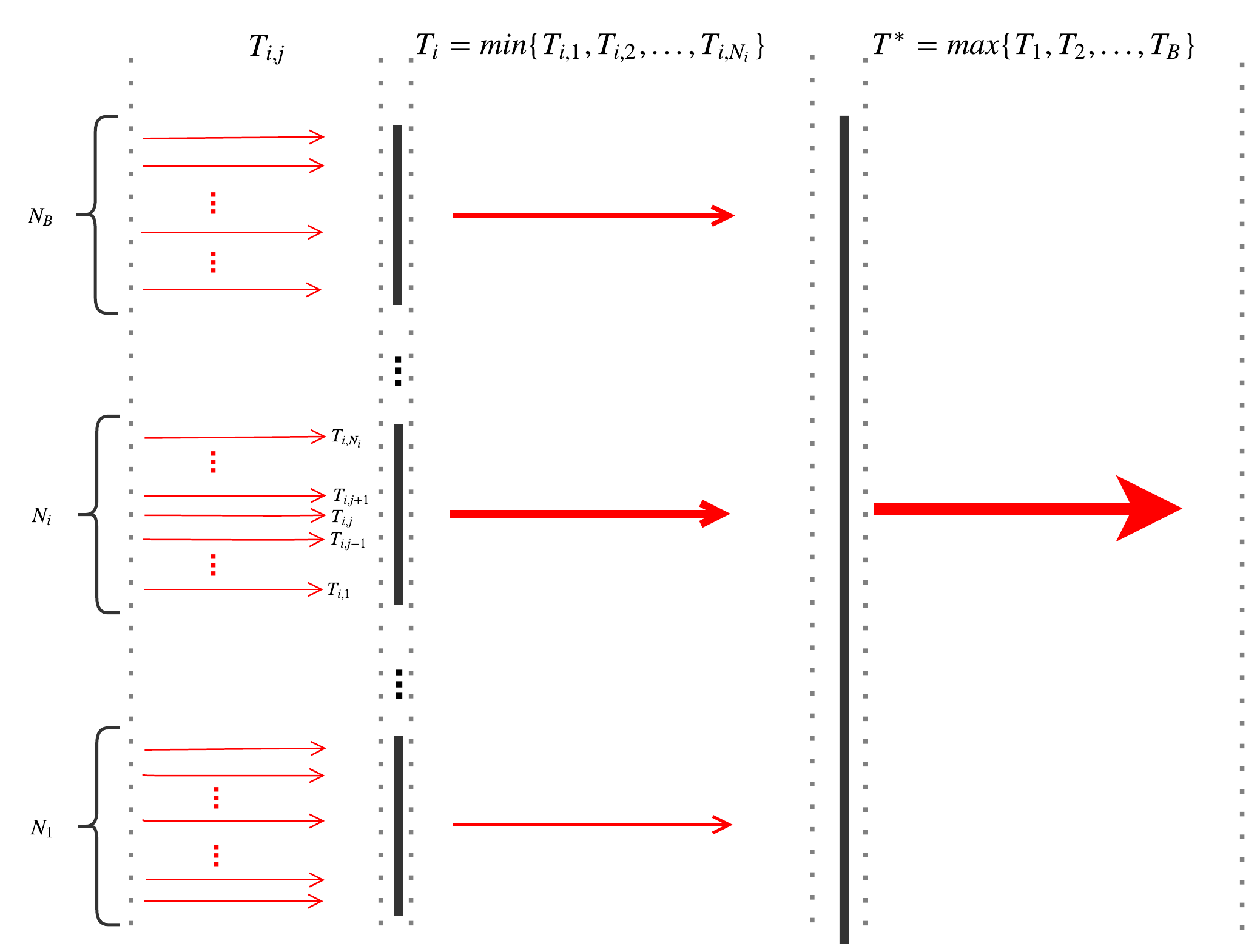}
   \caption{Time diagram of the system in Fig. \ref{fig:arc}.}
   \label{fig:fork}
\end{figure}

As illustrated in Fig.~\ref{fig:fork},  $N_i$ workers start computations over batch $i$ simultaneously. Thus, the result from the fastest worker, out of $N_i$, is sufficient for the master to recover the computation over batch $i$. Therefore, the $i$'th batch recovery time is the first order statistics of $N_i$ i.i.d exponential random variables, given by
\begin{equation}
    T_i=\min\left(T_{i,1},T_{i,2},\dots,T_{i,N_i}\right), \hspace{0.2cm} \forall i\in \{1,2,\dots,B\},
\label{min}
\end{equation}
where $T_{i,j}$s are i.i.d exponential random variables with rate $\lambda$, i.e.
$T_{i,j}\sim \exp(\lambda)$.
According to the exponential distribution properties \cite{desu1971characterization}, the recovery time of batch $i$ is also exponentially distributed with rate $N_i\lambda$,
\begin{equation}
    T_i\sim \exp\left(N_i\lambda\right),\quad \forall i\in \{1,2,\dots,B\}.
\end{equation}

For generating the overall result, the master node has to wait for the results of  computing over all batches. In other words, the overall result can be generated only after the slowest group of workers hosting the same batch deliver the local result. Hence, the overall result generating time, $T$, could be written as,
\begin{equation}
    T=\max\left(T_1,T_2,\dots,T_B\right).
\label{max}
\end{equation}
Here $T$ is the maximum order statistics of $B$ exponential random variables.

The question we address next is how should one redundantly assign $B$ non-overlapping batches of data among $N>B$ workers to achieve the shortest expected service time in System~\ref{fig:arc}. To this end, we next introduce several concepts.

\begin{definition}
The real valued random variable $X$ is greater than or equal to the real valued random variable $Y$ in the sense of \textbf{usual stochastic ordering}, shown by $X\underset{st}{\geq} Y$, if their tail distributions satisfy
\label{def1}
\begin{equation}
    P\{X>\beta\} \geq P\{Y>\beta\}, \quad \forall\beta\in \mathbb{R},
\end{equation}
or equivalently,
\begin{equation*}
   E[\phi(X)] \geq E[\phi(Y)],
\end{equation*}
for any non-decreasing function $\phi$.
\end{definition}

\begin{definition}
The random variable $X(\theta)$ is \textbf{stochastically decreasing and convex} if its tail distribution, $\Bar{F}_X(x) = P\{X>x\}$, is a pointwise decreasing and convex function of $\theta$.
\label{decreasing}
\end{definition}

\begin{definition}
For any $V_p=(v_{p1},v_{p2},\dots,v_{pM})$ in $\mathbb{R}^M$, the \textbf{rearranged coordinate vector} $V_{[p]}$ is defined as $V_{[p]}=(v_{[p1]},v_{[p2]},\dots,v_{[pM]})$, the elements of which are the elements of $V$ rearranged in decreasing order, i.e, $v_{[p1]}>v_{[p2]}>\dots>v_{[pM]}$.
\label{rearranged}
\end{definition}

\begin{definition}
Let $V_{[p]}=\left(v_{[p1]},v_{[p2]},\dots,v_{[pM]}\right)$ and $V_{[q]}=\left(v_{[q1]},v_{[q2]},\dots,v_{[qM]}\right)$  be two rearranged coordinate vectors in $\mathbb{R}^M$. Then $V_p$ \textbf{majorizes} $V_q$, denoted by $V_p\succeq V_q$, if
\begin{equation*}
\begin{split}
    &\sum_{i=1}^{m}v_{[pi]}\geq \sum_{i=1}^{m}v_{[qi]}, \enskip \forall m\in \{1,2,\dots,M\}, \enskip and\\
    &\sum_{i=1}^{M}v_{[pi]}= \sum_{i=1}^{M}v_{[qi]}.
\end{split}
\end{equation*}
\end{definition}

\begin{definition}
A real valued function $\phi:\mathbb{R}^M\rightarrow \mathbb{R}$ is \textbf{Schur convex} if for every $V$ and $W$ in $\mathbb{R}^M$, $V\succeq W$ implies $\phi(V)\geq \phi(W)$.
\end{definition}

\begin{definition}
A real valued random variable $Z(\Bar{x})$, $\Bar{x}\in \mathbb{R}^M$, is \textbf{stochastically schur convex} , in the sense of usual stochastic ordering, if for any $\Bar{x}$ and $\Bar{y}$ in $\mathbb{R}^M$, $\Bar{x}\succeq\Bar{y}$ implies $Z(\Bar{x})\geq Z(\Bar{y})$.
\end{definition}

\begin{lemma}
If the batch assignment $\EuScript{\Bar{N}}_{1}=(N_{11},N_{12},\dots,N_{1B})$
majorizes the batch assignment $\EuScript{\Bar{N}}_{2}=(N_{21},N_{22},\dots,N_{2B})$, that is, $\EuScript{\Bar{N}}_{1}\succeq\EuScript{\Bar{N}}_{2}$, then the corresponding service times $T(\EuScript{\Bar{N}}_{1})$ and $T(\EuScript{\Bar{N}}_{2})$ satisfy
\begin{equation*}
   E[ T(\EuScript{\Bar{N}}_{1})] \geq  E[ T(\EuScript{\Bar{N}}_{2})],
\end{equation*}
when the service time of the workers are independent and exponentially, identically distributed.
\label{lem}
\end{lemma}
\begin{proof}
The service time for batch assignment policy $\EuScript{\Bar{N}}_{k}$, $\forall k\in\{1,2\}$, is given by
\begin{equation*}
     T(\EuScript{\Bar{N}}_{k})=\max\left(T_{k1},T_{k2},\dots,T_{kB}\right),
\end{equation*}
where $T_{ki}$s are exponentially distributed with rate $N_{ki}\lambda$. Form Definition \ref{decreasing}, $T_{ki}$ $\forall i\in\{1,2,\dots,B\}$, is stochastically decreasing and convex function of $N_{ki}$. Therefore, $ T(\EuScript{\Bar{N}}_{i})$, which is a maximum of stochastically decreasing and convex functions, is stochastically decreasing and schur convex function of $\EuScript{\Bar{N}}_{i}$, \cite{liyanage1992allocation}. Hence, by definition,  $\EuScript{\Bar{N}}_{1}\succeq\EuScript{\Bar{N}}_{2}$ implies $T(\EuScript{\Bar{N}}_{1}) \geq  T(\EuScript{\Bar{N}}_{2})$ in the sense of usual stochastic ordering. Hence, for any non-decreasing function $\phi$,
\begin{equation*}
    E[\phi( T(\EuScript{\Bar{N}}_{1}))] \geq E[\phi( T(\EuScript{\Bar{N}}_{2}))].
\end{equation*}
Substituting $\phi$ by constant function completes the proof.
\end{proof}
\begin{proposition}
The balanced batch assignment, defined as $\EuScript{\Bar{N}}_{b}=(N/B,N/B,\dots,N/B)$ with $N$ and $B$ being the respective number of workers and batchs, is majorized by any other batch assignment policy.
\begin{proof}
See \cite{marshall1979inequalities}.
\end{proof}
\label{majorization}
\end{proposition}
Then Theorem (\ref{mainThm}) follows immediately from Lemma~\ref{lem}.
\begin{theorem}
With exponentially distributed service time of workers, among all (non-overlapping) batch assignment policies, the balanced assignment achieves the minimum expected service time for the overall result generation.
\label{mainThm}
\end{theorem}
\begin{proof}
From Proposition \ref{majorization}, $\EuScript{\Bar{N}}_a \succeq \EuScript{\Bar{N}}_b$ for any arbitrary batch assignment strategy $\EuScript{\Bar{N}}_a$ and from Lemma~\ref{lem}, and by choosing $\phi$ to be the constant function, the following inequality holds:
\begin{equation*}
E[ T(\EuScript{\Bar{N}}_{a})] \geq  E[ T(\EuScript{\Bar{N}}_{b})],
\end{equation*}
which means that among all non-overlapping batch assignment policies, the balanced assignment achieves the minimum expected computing time.
\end{proof}
Now, in the following proposition we precisely quantify the minimum expected computing time, which could be achieved by the balanced assignment.
\begin{proposition}
The lower bound of the expected time for overall result generation in System \ref{fig:arc}, with non-overlapping batches, when the service time of the workers are exponentially distributed with rate $\lambda$, is $\frac{B}{N\lambda}H_B$.
\end{proposition}
\begin{proof}
From (\ref{max}), the overall result generation time is the maximum order statistics of $B$ i.i.d random variables. With balanced assignment of non-overlapping batches and exponentially distributed service time of workers, $T_i$s are also exponential with rate $\frac{N\lambda}{B}$. On the other hand, the expected value of the maximum order statistics of $B$ i.i.d exponential random variables with rate $\mu$ is $\frac{1}{\mu}H_B$, \cite{crow2007two}. Substituting $\mu$ by $\frac{N\lambda}{B}$ completes the proof.
\end{proof}
\subsection{Computing Time with Overlapping Batches} \label{batching policy}
Recall the original problem of assigning a data set of size $S$ redundantly among $N$ workers. We assumed $S=N$. Each worker is assigned $N/B$  data blocks, for a given parameter $B$. Furthermore, the number of copies of each data block is identical and each worker hosts at most one copy of a block. Therefore, each data block will be available at exactly $N/B$ workers. 
\begin{figure*}%
    \centering
    \begin{subfigure}[Cyclic batching. Groups does not share same subset of the data set.]{%
        \label{6workers1}%
        \centering
        \includegraphics[width=5cm, trim={6cm 23cm 5cm 0},clip]{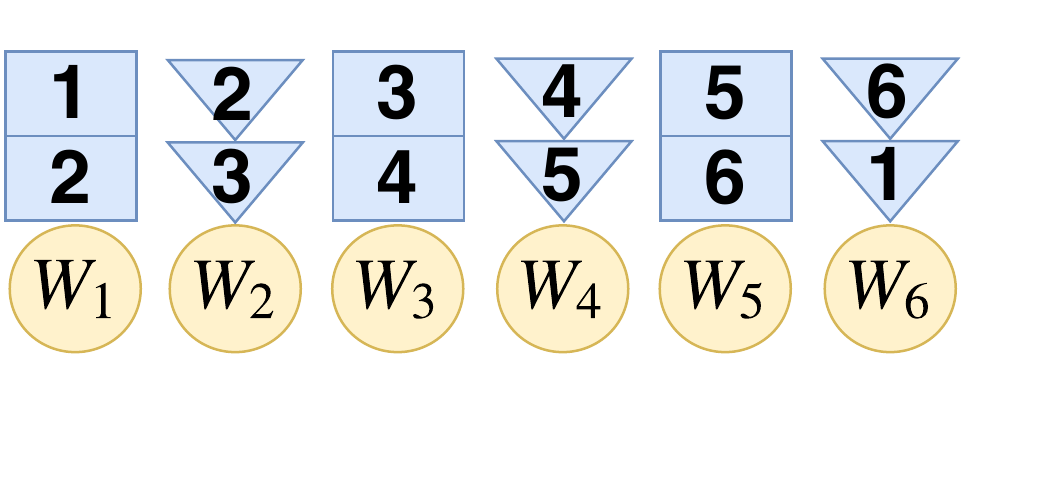}}%
    \end{subfigure}
    \quad
    \begin{subfigure}[Groups share a same subset of data set at $W_5$ and $W_6$.]{%
        \label{6workers2}%
        \centering
        \includegraphics[width=5cm, trim={6cm 23cm 5cm 0},clip]{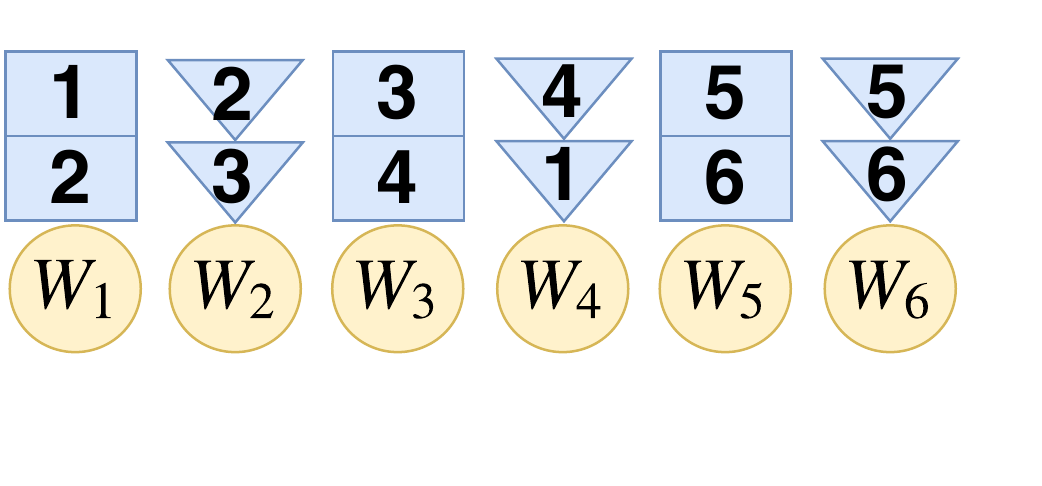}}%
    \end{subfigure}
    \quad
    \begin{subfigure}[Groups share same subsets of data set (non-overlapping batches).]{%
        \label{6workers3}%
        \centering
        \includegraphics[width=5cm, trim={6cm 23cm 5cm 0},clip]{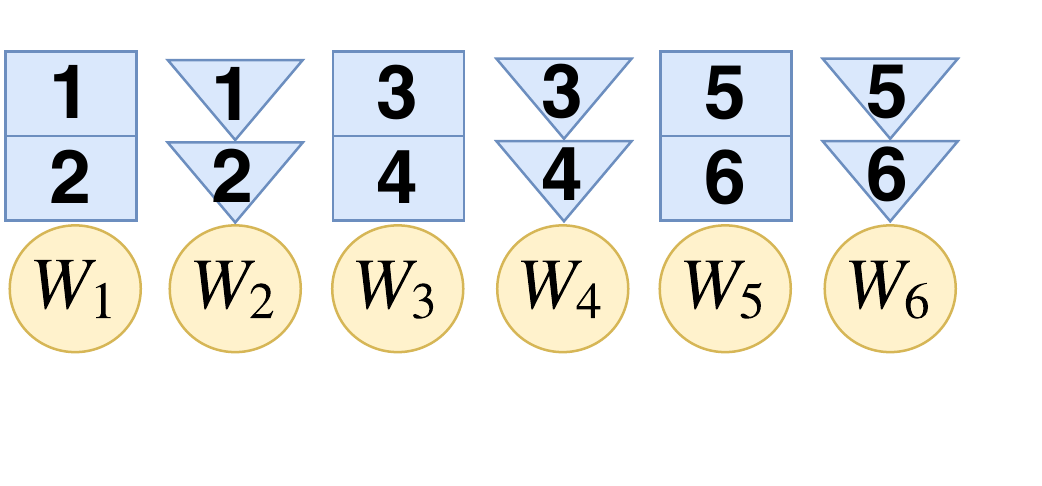}}%
    \end{subfigure}
    \caption{Overlapping batches with $N=S=6$ and $B=3$.}
    \label{6workers}
\end{figure*}

Suppose that the data set is divided into $N$ overlapping batches, each with size $N/B$, in a cyclic order, such that the first batch consists of blocks 1 through $N/B$, the second batch is comprised of blocks 2 through $N/B+1$, and so on. An example of this batching is given in Fig.~\ref{6workers1}. With this data batching policy, which we call it \textit{cyclic batching}, no two batches share the exact same data blocks but, on the other hand, the number of batches which share at least one data block with a specific batch is maximum. Specifically, with cyclic batching, each bach shares at least one common data block with $2\left(N/B-1\right)$ other batches. On the other hand, with non-overlapping batches, as it is shown in Fig.~\ref{6workers3}, each batch shares exactly $N/B$ data blocks with $N/B-1$ other batches. With any other batching policy each batch shares same data blocks with more than $N/B-1$ and less than $2\left(N/B-1\right)$ other batches, an example of which is given in Fig.~\ref{6workers2}. In what follows, we will provide analytic comparison of the computing time of System \ref{fig:arc}, with three different batching policies, provided in Fig.~\ref{6workers}.

Consider System~\ref{fig:arc}, with $N=S=6$, $B=3$, and three different batching policies, as shown in Fig.~\ref{6workers}. In each policy, there are two groups, shown by different shapes. Each group comprises the entire data set, in all three policies. However, the placement of blocks in groups is different across batching policies.  Let $X_i$ $\forall i\in\{1,2,\dots,6\}$ be the i.i.d random vector of workers service time. Let us assume, without loss of generality, that $W_1$ is the fastest worker, delivering its local results before the rest of the workers. Then the overall result generating time for policy \ref{6workers1} is
\begin{equation}
   T^*_{(a)}=\min\left(\max\left(X_3,X_5\right),\max\left(X_2,X_4,X_6\right)\right).
   \label{4a}
\end{equation}
For policy \ref{6workers2}, the overall result generating time could be written as,
\begin{align}
    T^*_{(b)}=\min( & \max(X_3,\min(X_5,X_6)),\\
       &
    \max\left(\max\left(X_2,X_4\right),\min\left(X_5,X_6\right)\right).
       \label{4b}
\end{align}
Comparing (\ref{4a}) and (\ref{4b}), it is easy to see that $E[T^*_{(b)}]<E[T^*_{(a)}]$, since,
\begin{equation*}
\begin{split}
    E[\max\left(X_3,\min\left(X_5,X_6\right)\right)]<E[\max\left(X_3,X_5\right)],
\end{split}
\end{equation*}
\begin{equation*}
\begin{split}
    &E[\max\left(\min\left(X_5,X_6\right),\max\left(X_2,X_4\right)\right)]<\\
    &\hspace{5cm}E[\max\left(X_2,X_4,X_6\right)].
\end{split}
\end{equation*}
On the other hand, for policy \ref{6workers3},
\begin{equation}
    T^*_{(c)}=\max\left(\min\left(X_3,X_4\right),\min\left(X_5,X_6\right)\right).
    \label{4c}
\end{equation}
In order to be able to compare the computing time of \ref{6workers3}, we rewrite (\ref{4b}) as follows:
\begin{equation}
    T^*_{(b)}=\max\left(\min\left(X_3,\max\left(X_2,X_4\right)\right),\min\left(X_5,X_6\right)\right).
    \label{4b2}
\end{equation}
The first argument of the outmost max functions in (\ref{4c}) and (\ref{4b2}) are compared as,
\begin{equation*}
    E[\min\left(X_3,X_4\right)]<E[\min\left(X_3,\max\left(X_2,X_4\right)\right)].
\end{equation*}
Therefore, $E[T^*_{(c)}]<E[T^*_{(b)}]$. Accordingly, the expected computing times of three batching policies in Fig. \ref{6workers} are compared as follows:
\begin{equation}
    E[T^*_{(c)}]<E[T^*_{(b)}]<E[T^*_{(a)}],
\label{comparison}
\end{equation}
which essentially means that, the balanced assignment of non-overlapping batches achieves the minimum expected computing time when compared to overlapping batch assignment.

\section{Numerical Results}\label{sectionV}
In this section we numerically evaluate the analytical results obtained in Sec.~\ref{sectionIV}. In particular, we compare the computing time of balanced non-overlapping-batch assignment with cyclic overlapping-batch assignment.

\begin{figure}[htb]
   \centering
   \includegraphics[width=\columnwidth, keepaspectratio]{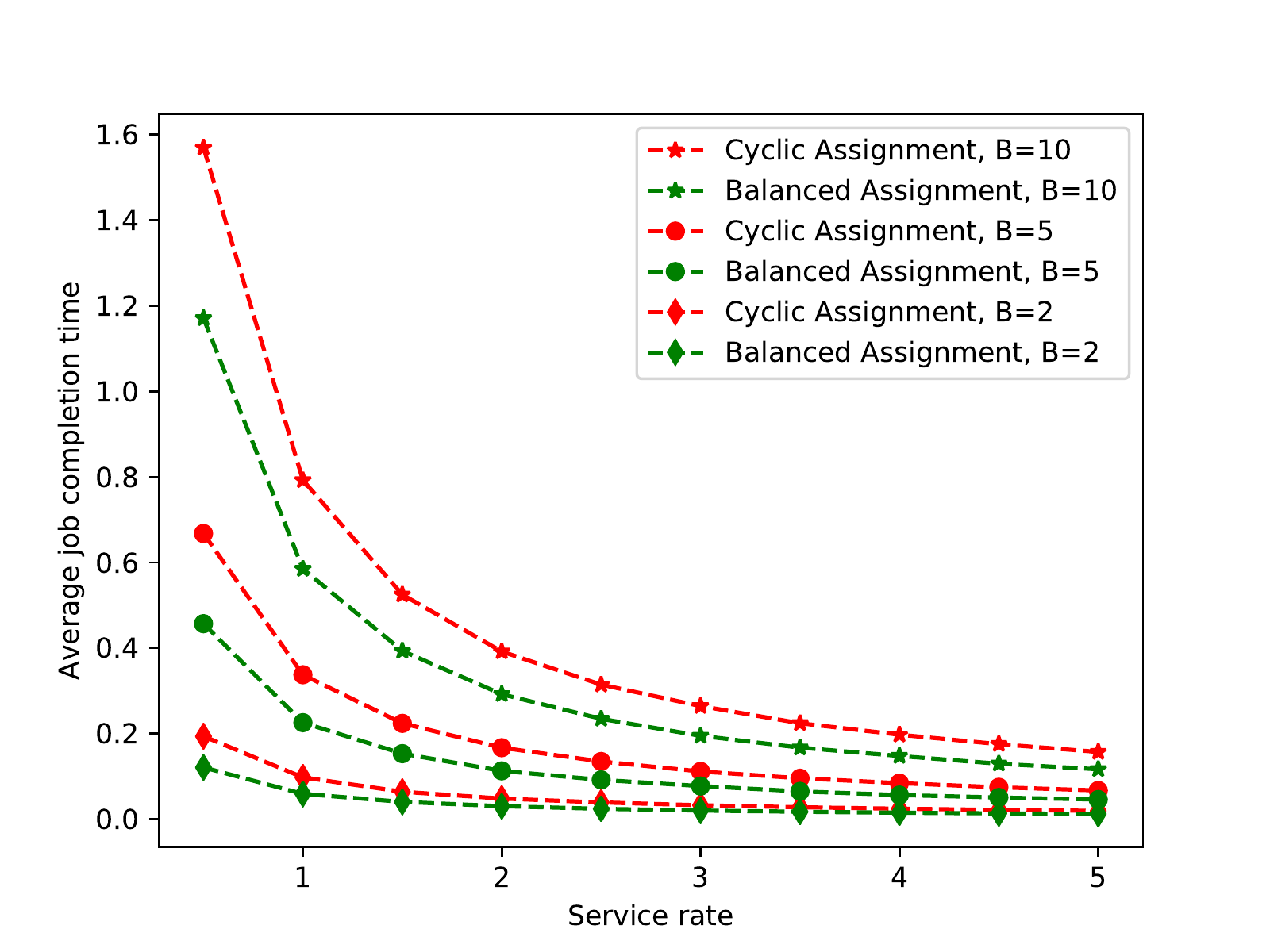}
   \caption{The comparison of the average job completion time for cyclic assignment \ref{6workers1} and balanced assignment \ref{6workers3}, with $N=50$ worker nodes and three different $B$ parameters.}
   \label{fig:sim1}
\end{figure}

In Fig.~\ref{fig:sim1}, the average computing time of balanced and cyclic assignments are plotted versus the service rate of the worker nodes for three different values of parameter $B$ and $N=S=50$. It can be seen that balanced assignment of the non-overlapping batches achieves lower computing time compared to cyclic assignment, for all three $B$s and service rates. This observation is in line with our earlier analytic result in (\ref{comparison}). Moreover, the gap between the two assignment polices gets larger in the low service rate region, which could be explained as follows. With exponentially distributed service times of workers, the time between each consequent local result deliveries to the master node is also exponentially distributed (due to the memoryless property of the exponential distribution), the average of which increases as the service rate of the workers decrease. On the other hand, from our analytic results we know that on average with the cyclic assignment, the master node would have to wait for more worker nodes to respond, compared to balanced assignment. Therefore, with lower service rate, the difference between computing times of the two assignment policies should increase. Besides, the performance gap decreases as the parameter $B$ decrease. In order to explain this behaviour, recall that the number of data blocks at each machine is $N/B$. Hence, for smaller $B$s, each worker does a larger portion of the computation and the master node should wait for a smaller number of worker nodes, which results in smaller performance gap between the two policies.

\section{Conclusion}\label{sectionVI}
The computing time of fault-tolerant distributed computing algorithms was analyzed. Specifically, the two recently proposed algorithms with focus on using error correcting codes and random task assignment were studied. It was analytically showed that, algorithms with the same amount of redundancy and failure tolerance would have different expected computing time. It was proved that, with a fixed number of computing tasks and exponentially distributed service time of the worker nodes, balanced assignment of non-overlapping data batches achieves minimum expected computing time, which was also validated by the simulation results.








\bibliographystyle{IEEEtran}
\bibliography{ref}

\end{document}